\documentclass[letterpaper,aps,pra,nolongbibliography,twocolumn,showpacs,floatfix,superscriptaddress,10pt]{revtex4-2} \usepackage{url}
\usepackage{amsmath}
\usepackage{amssymb}
\usepackage[utf8]{inputenc}
\usepackage[english]{babel}
\usepackage[T1]{fontenc}
\usepackage{amsmath}
\usepackage{hyperref}
\usepackage{lipsum}
\usepackage{graphicx,helvet}
\usepackage{color}

\usepackage{mathtools}

\usepackage{soul}
\usepackage{amsfonts}
 \usepackage{amsmath}
\usepackage{lipsum}\definecolor{mygray}{gray}{0.4}
\definecolor{light-blue}{rgb}{0.8,0.85,1}
\graphicspath{{images/}}
\usepackage{amsthm}
\usepackage[draft,inline,nomargin]{fixme} \fxsetup{theme=color} \FXRegisterAuthor{cv}{acv}{\color{magenta}CV}
 \FXRegisterAuthor{cp}{acp}{\color{blue}CP}
 \FXRegisterAuthor{dd}{ddd}{\color{red} DD}
 \FXRegisterAuthor{ar}{aar}{\color{green} AR}
\renewcommand{\>}{\rangle}
\mathchardef\Re="023C
\mathchardef\Im="023D
\usepackage[mathlines]{lineno}  
\newcommand{\<}{\langle}

\newcommand{\mcB}{\mathcal{B}}

\newcommand{\mcE}{\mathcal{E}}

\newcommand{\mcA}{\mathcal{A}}

\newtheorem{theorem}{Theorem}

\newcommand{\mcF}{\mathcal{F}}

\newcommand{\mcS}{\mathcal{S}}
\newcommand{\ie}{i.e.}
\newcommand{\fmlong}{fuzzy measurements}
\newcommand{\Fmlong}{Fuzzy measurements}

\newcommand{\rmi}{\mathrm{i}}

\newcommand{\mcC}{\mathcal{C}}
\newcommand{\cg}{\mcC}

\newcommand{\eref}[1]{eq.~(\ref{#1})} 

\newcommand{\fref}[1]{fig.~\ref{#1}}

\newcommand{\tr}{\mathop{\mathrm{Tr}}\nolimits}

\newcommand{\ket}[1]{{\vert #1 \rangle}}

\newcommand{\proj}[2]{{\vert #1 \rangle \langle #2 \vert}}

\newcommand{\unam}{Universidad Nacional Aut\'onoma de M\'exico, Ciudad de M\'exico 01000, M\'exico}
\newcommand{\ifunam}{Instituto de F\'{\i}sica, \unam}

\begin{document}
\title{\Fmlong{} and coarse graining in quantum many-body systems}
\author{Carlos Pineda} \email{carlospgmat03@gmail.com} \affiliation{\ifunam}
\author{David Davalos}
\affiliation{RCQI, Institute of Physics, Slovak Academy of Sciences, D\'ubravsk\'a cesta 9, 84511 Bratislava, Slovakia}
 \affiliation{\ifunam}
\author{Carlos Viviescas} \affiliation{Departamento de F\'{\i}sica, Universidad Nacional de Colombia, Carrera 30 No.~45-03, Bogotá D.C., Colombia}
\author{Antonio Rosado}\affiliation{\ifunam}
\begin{abstract} Using the quantum map formalism, we provide a framework to construct fuzzy and
coarse grained quantum maps of many-body systems that account for limitations
in the resolution of real measurement devices probing them. The first set of
maps handles particle-indexing errors, while the second deals with the effects
of detectors that can only resolve a fraction of the system constituents. We
apply these maps to a spin-$1/2$ $XX$-chain obtaining a blurred picture of the
entanglement generation and propagation in the system. By construction, both
maps are simply  related via a partial trace, which allows us to concentrate on
the properties of the former. We fully characterize the fuzzy map, identifying
its symmetries and invariant spaces. We show that the volume of the
tomographically accessible states decreases at a double-exponential rate in the
number of particles, imposing severe bounds on the ability to read and use
information of a many-body quantum system.

\end{abstract} 

\maketitle

\section{Introduction} Ultimately, like with all systems in nature, limitations in measurements set
the boundaries of what we can learn about the quantum world. Not surprisingly,
the measurement problem has played a prominent role in quantum mechanics since its
foundation~\cite{Zurek:99827}. Nowadays, quantum many-body systems can be
probed and manipulated at the single-particle level~\cite{Ott:2016dc}, allowing
for the study with unprecedented detail of problems ranging from quantum
entanglement to thermalization~\cite{RevModPhys.91.021001}. Such dazzling
advances in recent years have been possible due to an unparalleled
development of measurement techniques~\cite{Ott:2016dc, Bloch:2018kb} as well
as of our understanding of the measurement process in quantum
mechanics~\cite{Wise09,Jacobs:2014ui}. Yet, with the advent of quantum
technologies, technical and conceptual challenges remain.

Prominently, the scalability required for quantum technology places great
demands on the measurement and control of quantum systems with an increasing
number of degrees of freedom. At present, the necessary resources for the
experimental manipulation and characterization of many-body quantum systems at
the single particle level swiftly mount up with the size of the system
\cite{Haffner:2005hs, Preskill:2018gt}, making such detailed description
unfeasible. A technical obstacle that evinces the pressing necessity of an
accurate effective characterization of these systems when probed with the use
of imperfect measuring devices. In this work, invoking the language of quantum
channels, we construct a framework to study this emerging portrayal of many-body
systems.  

Imperfect measurements were considered by Peres, when assessing the
consequences of having an experimental uncertainty in the eigenvalue
measurement larger than the difference between consecutive
eigenvalues~\cite{peres1995quantum}. In this context, projective measurements
were generalized to positive operator valued measurements, allowing for a more
gentle effect on the measured system~\cite{PhysRevD.20.384} and incorporating
unsharp, weak, and fuzzy measurement
schemes~\cite{GUDDER200518,carmeli2008,Busch2010}.
More recently, as an alternative to the decoherence program, in
\cite{Brukner2007} and \cite{Brukner2008} it has been argued that classical features
emerge from coarse grained descriptions of quantum systems when measured with
imperfect apparatuses, sparking research on \emph{ad hoc} models to estimate
the resilience of quantum features under imperfect measurements
\cite{Duarte2017, PhysRevA.100.022334, Carvalho:2020fo}, arguably, an
additional issue to account for in the development of quantum
technologies~\cite{Preskill:2018gt}.

This progress notwithstanding, a solid framework that can systematically
address the implications of studying quantum many-body systems using imperfect
measurements is still lacking. This work aims at closing this gap. Our
formalism considers imperfect detection of two different, but deeply related,
kinds in quantum many-body systems: \emph{fuzzy measurements} (FMs), in which
single particles can be resolved, however, there is always a finite probability
of their misidentification; and \emph{coarse graining} (CG) measurements, in which
groups of particles are treated collectively as an effective particle, either
because the measurement device is not sensible enough to resolve all degrees of
freedom in the system, or because only a few of these degrees of freedom are of
relevance for the system property being studied.  We formalize both concepts in
the language of quantum maps, identifying their symmetries, spectra, and
invariant spaces. Physical consequences of our results are illustrated with
spin-entanglement waves, in which a blurring of the observable entanglement is
observed, and a study of the contractive properties of our maps, showing a
double-exponential contraction rate of the accessible state volume with the
number of particles, a remarkable result hinting at the fragility of quantum
resources with respect to imperfect measurements.

 \section{\Fmlong{}} Consider the situation in which a single-particle measurement is performed on a many-body system, but one is not sure on
which particle this measurement was applied. For example, one shines an ion
chain, and obtains a fluorescent signal. However, due to the addressing
imperfections of the detector device, one is not able to determine the exact
origin of the fluorescent signal. The obtained information in this case becomes
blurred, yet its quantification is still possible. \par

 \par
Consider first the simplest many-body system: two particles. Suppose a measurement of the observable $A\otimes B$ in a two particle system
is wanted. Yet, with probability $1-p$, the measurement apparatus mistakes the
particles, so instead, sometimes, a measurement of $B\otimes A$ is done. Hence,
if $\rho$ is the state of the system, the outcome of this FM is
\begin{equation}
p\tr \rho A\otimes B + (1-p) \tr \rho B\otimes A
= \< A\otimes B\>_{\mcF_{2\text{p}}[\rho]},
\label{eq:ejemplo:fuzzy:dos}
\end{equation}
with $\mcF_{2\text{p}}[\rho]=p \rho +(1-p)
S_{01}[\rho]$, where the brackets in front of a unitary operator denote its
natural action on density matrices, e.g. $S_{ij}[\rho] = S_{ij} \rho
S_{ij}^{\dagger}$ is the application of the swap gate with respect to particles
$i$ and $j$ to the system state. Incidentally, note that if the error is
assumed to be the same for every observable of the system, then
$\mcF_{2\text{p}}[\rho]$ corresponds to the tomographically accessible state
\footnote{Consider $\left\{A_i \right\}_i$ to be a set of 
tomographically complete observables, suffering from the fuzzy noise. Then the measurable
averages are {$\tr((p A_{i} + (1-p) S A_{i} S) \rho) =  \langle A_i
\rangle_{\mcF_{2\text{p}}[\rho]}$}  for all $A_i$. Since {$\langle A_i
\rangle_{\mcF_{2\text{p}}[\rho]}$} are the components of
$\mcF_{2\text{p}}[\rho]$ in the complete operator basis $\left\{A_i \right\}$,
then $\mcF_{2\text{p}}[\rho]$ is the state accessible tomographically.}.

The simple reasoning above can be followed for an $n$-body system. If
the measurement device wrongly identifies particles with probability $p_{P}
\ge 0$, according to permutation $P$, the outcome of the FM of
operator $M$ is $\tr \left(M \mcF[\rho]\right)$, with
\begin{equation}
\mcF[\rho] = \sum_{P\in \mcS} p_{P} P[\rho], 
\label{eq:fuzzy_general}
\end{equation}
where $\mcS$ is a subset of the symmetric group of $n$ particles, and $\sum_{P \in
\mcS} p_{P} = 1$. We regard Eq.~\eqref{eq:fuzzy_general} as the most general
form of the FM channel and consider some particular examples below.

Quite generally, during a measurement, two-body errors are more likely to occur
than higher order ones. Accounting just for these, the effective
tomographically accessible state is
\begin{equation}
\mcF_{2\text{b}}[\rho] = p \rho + (1-p)\sum_{i,j} p_{ij} S_{ij}[\rho],
\label{eq:fuzzy_two_body}
\end{equation}
with $\mcS$ containing only swaps.  A relevant related case is a periodic
one-dimensional chain in which errors between adjacent particles are the only
ones taken into account. In this case  \begin{equation}
\mcF_{1\text{d}}[\rho] = 
    p \rho + (1-p)\sum_i p_i 
S_{i,i+1}[\rho]
\label{eq:fuzzy_chain}
\end{equation}
with $\mcS$ consisting of neighboring exchanges.

Similar considerations can be done for two-dimensional systems or more refined
proposals, e.g. a swap probability that decays with the distance between $i$ and
$j$, or a particular experimental setup.
\par
\section{Coarse graining}  Consider now a slightly worse situation in which the measurement device is
unable to resolve the fine details of the whole system and is forced to capture
effective
reduced states of randomly chosen subsets of the system. Two processes
characterize such an apparatus: the choice of random subsets
and the reduction to a representative state as a result of a
partial trace of the subset. The latter accounts for the discarded information of the
inaccessible parts and provides a \emph{coarse-grained} picture of the system
state. This is illustrated in the simple case of a two-particle system.
The expected value of the single-particle observable $M$ measured by an
apparatus that with probability $p$ ($1-p$) detects the first (second) particle
is  $\tr \left(M \rho_\text{eff}\right)$, with the effective single-particle
state $\rho_\text{eff}=\tr_1 \mcF_{2\text{p}}[\rho]$ [cf.
\eref{eq:ejemplo:fuzzy:dos}]. 
\par

The same procedure can be applied to systems with more than two particles.
Assume that a detector is able to measure only $N$ particles at a time, picking each
one of them randomly from subsets of $m_{k}$ particles, with $k=1,\dots, N$. The effective $N$-particle state is given by 
\begin{equation}
\rho_\text{eff} := 
\left(
\bigotimes_{k=1}^N \cg_{m_k}
\right)[\rho]
\end{equation}
with
\begin{equation}
\cg_{m}[\rho] = \frac{1}{m} \tr_{\overline{1}} \sum_{i=1}^m S_{1,i}[\rho],
\label{eq:bonche}
\end{equation}
where the overline in the subindex of the trace denotes the set complement, in this case trace over all but the first particle.
As we now show, this channel is the composition of a
FM and a partial trace: a
consequence of the linearity of the trace and the group properties of
permutations. Let us start by factorizing the partial traces in a tensor product,
as their domains are disjoint. Noting that $\mcF^{(m)}[\rho] := 1/m
\sum_{i=1}^m S_{1,i}[\rho]$ is a fuzzy map, we can write
\begin{equation}
\label{eq:kappaF}
\begin{split}
\bigotimes_{k=1}^N \cg_{m_k} &=\tr^{(1)}_{\overline{1}}\otimes \dots \otimes \tr^{(N)}_{\overline{1}} \left( \mcF^{(m_1)}\otimes \dots \otimes \mcF^{(m_N)} \right) \\
&= \tr_{\overline{\left\{ 1_1, \dots, 1_N \right\}} } \left( \mcF^{(m_1)}\otimes \dots \otimes \mcF^{(m_N)} \right).
\end{split}
\end{equation}
Now, by distributing the tensor product, we can define the FM
\begin{align*}
\mcF &:=  \mcF^{(m_1)}\otimes \dots \otimes \mcF^{(m_N)}  \\
&= \left(\frac{1}{m_1}\sum_{i_1=1}^{m_1} S_{1,i_1}\right) \otimes \dots \otimes \left(\frac{1}{m_N}\sum_{i_N=1}^{m_N} S_{1,i_N} \right)\\
&= \prod_{j=1}^N \frac{1}{m_j} \sum_{i_1,\dots,i_N=1}^{m_1,\dots,m_N} S_{1,i_1}\otimes \dots \otimes S_{1,i_N},
\end{align*}
since tensor products of swaps result in disjoint permutations with respect to
particle sets. We can now rewrite the effective $N$-particle state as the
result of applying the CG map to $\rho$,
\begin{equation}
\rho_\text{eff} := \tr_\kappa \mcF[\rho],
\end{equation}
where $\kappa$ is a suitable choice of subsystem to
trace, as given in \eqref{eq:kappaF}. The case $m_{k}=m=3$  is depicted in
\fref{fig:cg_chain}(a).

\par Motivated by a different physical feasible situation, we examine a second
CG measurement with its corresponding FM. Consider a
chain of particles such that the conditions with respect to the measurement
device alternate, e.g. the even ones are closer to the detector than the odd
ones; see~\fref{fig:cg_chain}(b) for an example. Assume
that with probability $p$ the measurement device measures the even particles,
but with a position dependent probability one of its odd-labeled neighbors is
detected. 
The channel associated with such a measurement is
\begin{equation}
   p\tr_\text{odd}\rho 
   + (1-p) \tr_\text{odd} \sum_{i} p_i S_{i,i+1}[\rho ] 
=\tr_\text{odd}\mcF_{1\text{d}}[\rho].
\nonumber
 \end{equation}
Single particle observables of, say, particle $2i$ are evaluated with respect
 to the effective
 density matrix  
 $[1 +p(p_{2i-1}+p_{2i})] \rho_{2i} + pp_{2i-1}\rho_{2i-1} + p p_{2i} \rho_{2i+1}$, where $\rho_i$ is the reduced density matrix for particle $i$. 
 Two particle observables, say of particles $2i$ and $2j$, are calculated with a density 
 matrix, with coefficients of order $0$ in $1-p$ of  $\rho_{2i,2j}$, and other
 contributions of order $1$, of $\rho_{2i\pm1,2j}$ and  $\rho_{2i,2j\pm1}$.
\par \begin{figure} \centering
\includegraphics[width=\columnwidth]{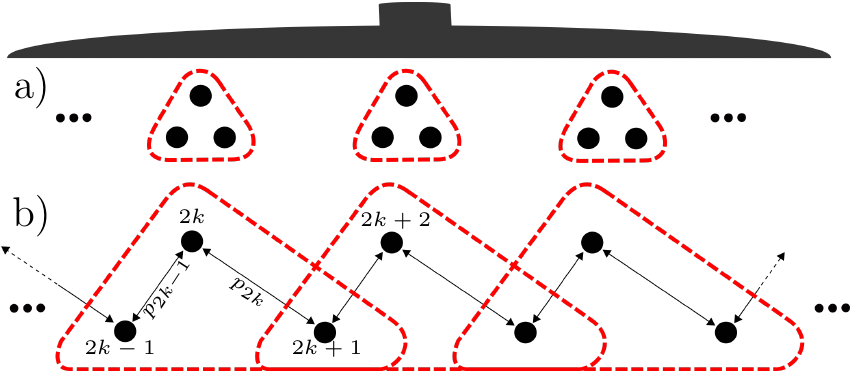}
\caption{CG schemes. (a)
The detector groups sets of $m=3$ particles into one; see \eref{eq:bonche}. 
(b)
The detector measures even particles with a higher probability than odd ones, but
sometimes it
mistakes them for one of their neighbors.
\label{fig:cg_chain}}
\end{figure} The above examples illustrate how to construct reduced quantum states using quantum
maps motivated by probabilistically chosen subsystems, based on the 
physical arrangement of both the system and the measurement device;
this construction might include the simultaneous exchange of a greater number
of particles. 
An encompassing scheme for CG maps is
\begin{align}
\cg[\rho]
&=
\tr_\tau 
\sum_{P \in \mcS} p_{P} 
P \rho P^\dagger{} =\tr_\tau \mcF[\rho],
\label{eq:fuzzy_cg_general_relation}
\end{align}
where $\tau$ denotes the part of the system that is traced.

 \section{Fuzzy and coarse grained entanglement waves}To illustrate the use of the proposed maps, we now consider the recently achieved observation of single spin impurity
dynamics \cite{Fukuhara:2013hq} and spin-entanglement wave propagation
\cite{Fukuhara:2015ec} in one-dimensional Bose-Hubbard chains at the level of
single-atom-resolution detection \cite{Sherson:2010hg,Bakr:2010gd}.\par The dynamics of a single spin impurity in a one dimensional homogeneous
spin-$1/2$ $XX$-chain is generated by the  Hamiltonian $\hat{H} = -
J_{\text{ex}} \sum_{j} (\hat{\sigma}^{+}_{j} \hat{\sigma}^{-}_{j+1} +
\hat{\sigma}^{-}_{j} \hat{\sigma}^{+}_{j+1}) $, where $J_{\text{ex}}$ is the
exchange coupling and $\hat{\sigma}^{\pm}_{j}$ are spin-$1/2$ raising (lowering)
operators acting on particle $j$~\cite{Subrahmanyam:2004eo,Amico:2004ck}.  
We write an infinite spin-chain state with a single spin-up
impurity at site $j$ as $\ket{j} \equiv \ket{\dots, 0_{j-1}, 1_{j}, 0_{j+1},
\dots}$, where $\ket{1}$ ($\ket{0}$) refers to spin a spin-up (spin-down) state in the
$z$ basis, and choose the initial state as $\ket{0}$, \ie{}, with the spin-up impurity
at the center of the chain. At later times the excitation travels coherently to both sides of the chain as described by the state of the system
$\ket{\psi (t)} = \sum_{j} \phi_{j}(t) \ket{j}$, where $\phi_{j}(t) = \rmi^{j}
J_{j} (t J_{\text{ex}}/ \hbar)$ with $J_{j} (x)$ the Bessel function of the
first kind~\cite{Konno:2005da}.
In the first column in \fref{fig:chess}, the generation of concurrence and
its spread in a wave-like fashion along the
spin chain is shown~\cite{Amico:2004ck,Subrahmanyam:2004eo,Mazza:2015km}. 
It is noticeable how the maximum of the concurrence propagates to neighboring sites, moving away from the center.
\par We now study the entanglement in the system under FM. 
Assume, that with probability $1-p$, the measurement apparatus
is displaced equiprobably one site to the right or to the left of the chain. This case
is described by \eref{eq:fuzzy_general} with $\mcS=\{\openone, T, T^\dagger\}$,
and $T|j\>=|j+1\>$, $p_{\openone} =p$, $p_T=p_{T^\dagger} = (1-p)/2$. As shown
in the second column of \fref{fig:chess}, the entanglement still exhibits its
wavelike spreading albeit with a weaker intensity compared to the unaltered dynamics.
This decoherence effect of the FM is
responsible for the blurry appearance of the images, which however is not
homogeneous, and is most prominent among distant symmetric pairs.
Observe, for instance, that the squarelike structures that can be seen 
for the exact dynamics at $t=6$ are almost completely lost when the FM is applied.
\par We also consider the effects of CG. For this, we group the
particles in disjoint  sets, as in \eref{eq:bonche}. To respect the symmetry $j\to -j$ of
the system, we group particles in sets of two ($\pm \{1,2\},\pm
\{3,4\},\ldots$), with $\mcC_2$, and sets of four  ($\pm \{1,2,3,4\},\pm
\{5,6,7,8\},\ldots$), with $\mcC_4$,  leaving particle $j=0$ unaltered. The
entanglement evolution for these two cases is shown in columns $3$ and $4$ in
\fref{fig:chess} for
$\mcC_2$ and $\mcC_4$, respectively. In both cases, despite the lower
resolution, the entanglement wavelike propagation can still be seen, with a
higher intensity for $\mcC_2$ than for $\mcC_4$. Interestingly, entanglement
seems to be suppressed more for neighboring particles than for distant
CG particles. This can be appreciated in the last two columns in
\fref{fig:chess}, where the pattern is concentrated along the diagonal.
\par \begin{figure} \begin{center}
\includegraphics[width=\columnwidth]{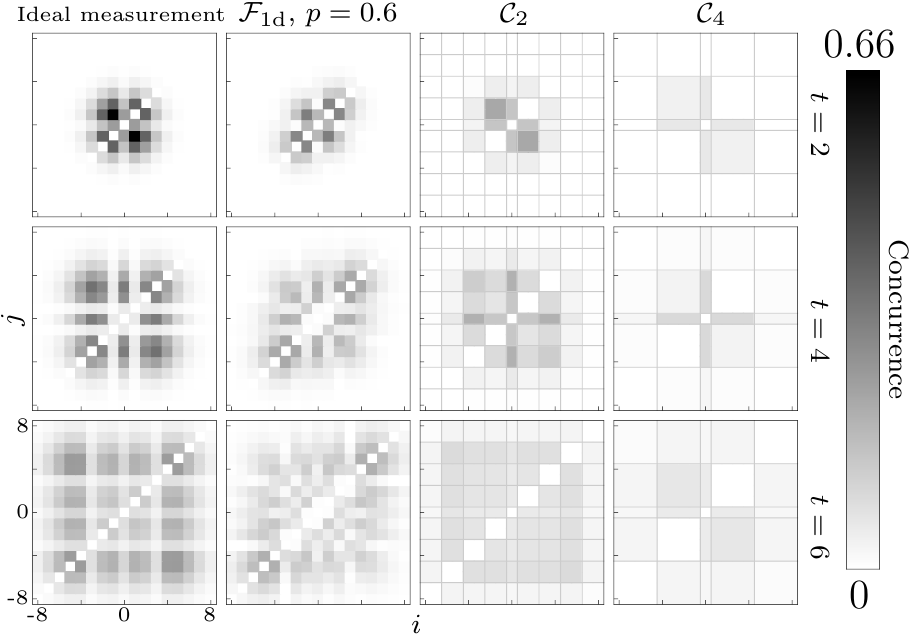}
\caption{
Concurrence for the single spin impurity
dynamics as a function of lattice sites $i$ and $j$ at different times. 
The first column shows the entanglement for the unaltered states. The second 
column 
corresponds to an FM, whereas the third and fourth columns display results
for coarse-grained descriptions,  grouping two and four particles, respectively
(gridlines group the coarse-grained sets of particles). The central particle
has not been coarse grained in order to maintain the symmetry.}
\label{fig:chess}
\end{center}
\end{figure}

 \section{Symmetries and spectrum}
\subsection{Symmetries of \fmlong{}} 
Understanding the symmetries of a physical system is crucial to understanding 
its dynamics. For closed quantum systems symmetries result in quantum numbers, 
which are used to understand and classify the system's spectrum.
Here, we calculate all symmetries for generic FMs, \ie{}, with $p_P> 0 $ for every permutation $P$ in the symmetric group,
and, as an important example, connect them with the generic invariant space. \par 
Consider a system of $n$ particles with a single-particle space of dimension $d$.
We start by introducing the linear superoperator
$\Gamma_{l,l'}[\varrho]=\gamma_{l,l'}\varrho$, which counts how many particles
in $\varrho$ are in $|l\> \< l'|$.
For example, for $\varrho = |00\>\<01|=  |0\>\<0|\otimes  |0\>\<1|$, 
$\Gamma_{0,0}[\varrho] = \Gamma_{0,1}[\varrho] =\varrho $ while 
$\Gamma_{1,0}[\varrho] = \Gamma_{1,1}[\varrho] =0 $.
Clearly, operators $\Gamma_{l,l'}$ commute among themselves. In addition, for
any particle permutation, $\Gamma_{l,l'} P[\varrho] = P \Gamma_{l,l'} [\rho]$,
and therefore they also commute with the FM, $[\mcF,\Gamma_{l,l'}] = 0$. Hence
$\mcF$ can be diagonalized in blocks $\mcF_\gamma$, each indexed by a
matrix $\gamma$ with $d\times d$ integers ranging from 0 to $n$. The subspace of physical states is the one labeled by all diagonal $\gamma$s. The number of blocks can be counted using the stars and bars theorem~\cite{starsbars} and is given by the binomial coefficient $\binom{d^2+n-1}{n}$.
Since the total number of particles is fixed, $\sum_{l,l'} \gamma_{l,l'} =n$ in
all blocks. \par

Note that there are some equivalent blocks. Working in the computational
basis, it can be shown by explicit substitution that if
$\Gamma_{l,l'}[\varrho]=\gamma_{l,l'}\varrho$, then 
$\Gamma_{l',l}[\varrho^\text{T}]=\gamma_{l',l} \varrho^\text{T}$, and that 
\begin{equation}
(\varrho_1,
\mcF[\varrho_2]) = (\varrho_1^\text{T} , \mcF[\varrho_2^\text{T}]),
\label{eq:identification_transpose}
\end{equation}
 for arbitrary $\varrho_1$ and $\varrho_2$. Both identities combined lead to
$\mcF_\gamma= \mcF_{\gamma^T}$ when an appropriate order of the computational
basis is used.  
Further block equivalences follow from relabeling symmetries of the levels that
allow block identification. Let $Q$ be any of the $d!$ unitary operators that
relabels the elements of the computational basis; for qubits this set
is $\{\openone, \sigma_x^{\otimes n} \}$. Since $[Q,P]=0$, $\forall P\in
\mcS$, it follows that 
\begin{equation}
(\rho_1, \mcF[\rho_2] ) = (Q\rho_1 Q', \mcF[Q \rho_2
Q']),
\label{eq:identification_Q}
\end{equation}
 for arbitrary $\varrho_1$ and $\varrho_2$, establishing the equivalence
between block $\mcF_{M(Q)\gamma
M(Q')^\text{T}}$ and block $\mcF_{\gamma}$, with $M(Q)$ the permutation matrix
of $d$ elements associated with $Q$.

Examining the qubit case provides some intuition in the interpretation 
of matrix $\gamma$. For qubits, $\gamma$ is a $2 \times 2$ matrix whose integer, 
semipositive entries should add up to $n$. This leaves three free parameters, which 
we organize as follows: $\alpha = \gamma_{1,0} + \gamma_{1,1}$ and 
$\beta = \gamma_{0,1} + \gamma_{1,1}$ which count the number of excitations
in the ket and in the bra, respectively, and $\gamma_{1,1}$. 
Thanks to the identification of blocks, we can always choose $\alpha,\beta,
\gamma \le n/2$, and the degeneracy $\delta$ depends on the repeated values
of the entries of $\gamma$. Thus, we can write
\begin{linenomath*}
\begin{equation*}
\mcF^\text{qubits} = 
\bigoplus_{\alpha,\beta,\gamma_{1,1} =0}^{ \lfloor n/2 \rfloor } 
     \mcF_{(\alpha,\beta,\gamma_{1,1})}^{\oplus \delta(\gamma)}.
\label{}
\end{equation*}
\end{linenomath*}
For the generic case the blocks $\mcF_\gamma$
are irreducible. To prove it, note that if all elements of the permutation group in $\mcF$
have a positive weight, all matrix elements of $\mcF_\gamma$ in the
computational basis of the corresponding subspace are strictly positive. This
is a consequence of the fact that
for any two such basis elements $\varrho_{1,2}$, characterized by the same
$\gamma$, there exists a permutation $P$ such that $P[\varrho_1] = \varrho_2$;
see Appendix~\ref{app:connectiongamma}.  Therefore the blocks $\mcF_\gamma$ are
irreducible \cite{Meyer2001}. Moreover, according to the Perron-Frobenius theorem~\cite{perron}, each
one contains only one invariant matrix. This implies that for diagonal
$\gamma$, the blocks $\mcF_\gamma$ define \textit{ergodic} quantum
channels~\cite{ergodic}.
\par
\subsection{Invariant space of \fmlong{}} Due to the Perron-Frobenius theorem and the total of irreducible blocks in the
generic case, it follows immediately that the dimension of the invariant space
is $C(d^2+n-1, n)$. Furthermore, using
Theorem $1$ in Appendix~\ref{app:connectiongamma}, it is easy to prove that for every
$\gamma$ the homogeneously weighted combination of all matrix elements of
$\mcF_\gamma$, in the computational basis, is an invariant matrix.  Thus such
matrices are the only invariant states per block (up to a scalar). In fact, such
matrices are invariant under any permutation according to the same theorem,
thus forming the symmetric set of matrices. In summary,
for the generic case and $\Delta$ a linear operator,
\begin{equation}
\mcF[\Delta]=\Delta \Leftrightarrow P[\Delta]=\Delta, \ \ \forall P\in \mcS.
\label{eq:invariancce_generic}
\end{equation}
Restricting $\Delta$ to be a positive definite operator, using Bellman
inequalities, a similar result which encompasses subgroups of the symmetric group, can be obtained for non-generic FMs such 
as $\mcF_{2\text{b}}$ and $\mcF_{1\text{d}}$; see
Appendix~\ref{app:invariant}. 
In conclusion, for FMs
whose permutations generate the symmetric group of $n$ particles, we identify
symmetric states as the only invariant ones. This set contains the
pure symmetric states which are simply eigenstates of the total angular
momentum [each particle being a spin-$(d-1)/2$] with the maximum eigenvalue.
As shown below, the set of non-symmetrical states is dramatically
affected by the action of the FM.
\par
\subsection{Volume contraction} \begin{figure} \centering
\includegraphics[]{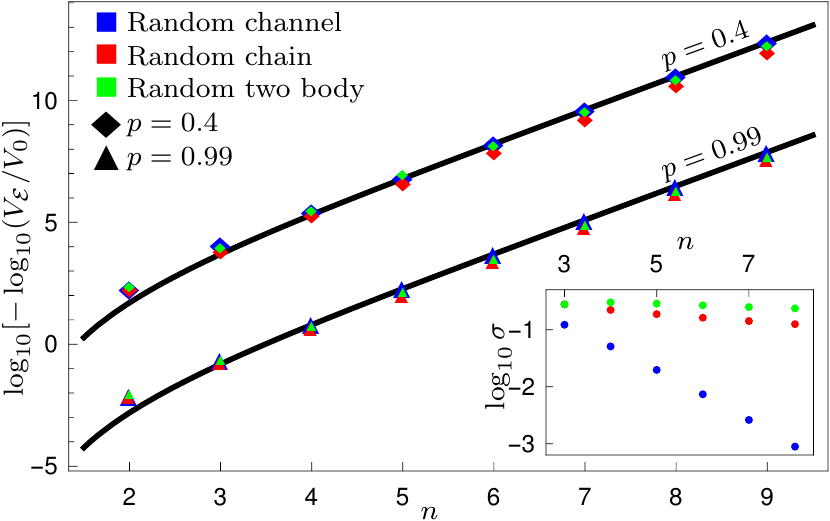}
\caption{
Volume contraction of the space of density matrices under a single
realization of random, two-body, and chain FMs, together 
with the ansatz,~\eref{eq:azats:contraction}. 
Double-exponential contraction of the state space is shown.
Inset: The standard deviation of the corresponding spectra (ignoring the
ones). Self-averaging is observed only for the
random general case, \eref{eq:random}.
\label{fig:sigma:gamma}}
\end{figure} In order to calculate the volume change due to the application of channel $\mcE$, we consider the manifold in which density matrices
live. All density matrices can be expressed as
$\openone/\tr\openone + \sum_{i=1}^D \alpha_l G_l$, where the 
$G_l $ form a traceless complete orthonormal set of matrices
and $D=n^d-1$. 
Treating $\alpha$ as a real vector in $\mathbb{R}^{D}$, we can consider
the volume (with the usual measure) in such $D=n^d-1$ space. For the qubit case
this volume corresponds to the volume in the Bloch sphere representation. 
Consider now the hypercube defined by the points 
$\rho$ and $\rho + \epsilon G_l$, with $l=1,\dots,D$. The volume is
$V_0 = \epsilon^D$. Under the map, the hypercube will transform 
to the $D$-parallelotope defined by the corners
$\mcE(\rho)$ and $\mcE(\rho + \epsilon G_l)$ whose signed volume is 
$V_\mcE = V_0 \det \mcE = V_0 \prod_l \lambda_l$, where $\lambda_l$ are the
eigenvalues of $\mcE$. 

We calculate the volume contraction for several of the channels presented in
this paper. 
First, consider the case of the random CG channel defined in 
\eref{eq:random}. Recall that 
the spectrum has an eigenvalue of 1 degenerated $C(d^2+n-1,n)$ times. 
We assume that the other eigenvalues are $p$, based  on self-averaging and
the spectral gap observed. This leads
to a volume contraction
given by 
\begin{equation}
\frac{V_\mcE}{V_0} 
\approx 
p^{\left[d^{2n} -{\textstyle \begin{psmallmatrix}d^2+n-1\\n \end{psmallmatrix}}\right]}. 
\label{eq:azats:contraction}
\end{equation}
Note that this is a double exponential in the number of particles. 
In \fref{fig:sigma:gamma} we show a comparison between this approximation and a single 
realization of the channel for a varying number of qubits. 

This implies that exploring the space of states, even with very efficient
detectors, is extremely difficult. This is coherent with the difficulties  found
in quantum tomography~\cite{Banaszek_2013,Flammia_2012}.  Fortunately, the
fraction of the Hilbert space in which nature lives, seems to be much
smaller~\cite{Orus2019,2009arXiv0905.0669P}. \par

Due to Uhlmann's theorem~\cite{UHLMANN1976273}, fuzzy states are majorized by
exact ones, \ie{} $\mcF[\rho]\prec \rho$. This means that any convex function
of the density matrices, such as the
von Neumann entropy~\cite{geometry2017}, a wide class
of entanglement measures~\cite{Horodecki} and other quantum
correlations~\cite{KUMAR20163044}, 
will
decrease under FM. 
\par
\begin{figure} \centering
\includegraphics[]{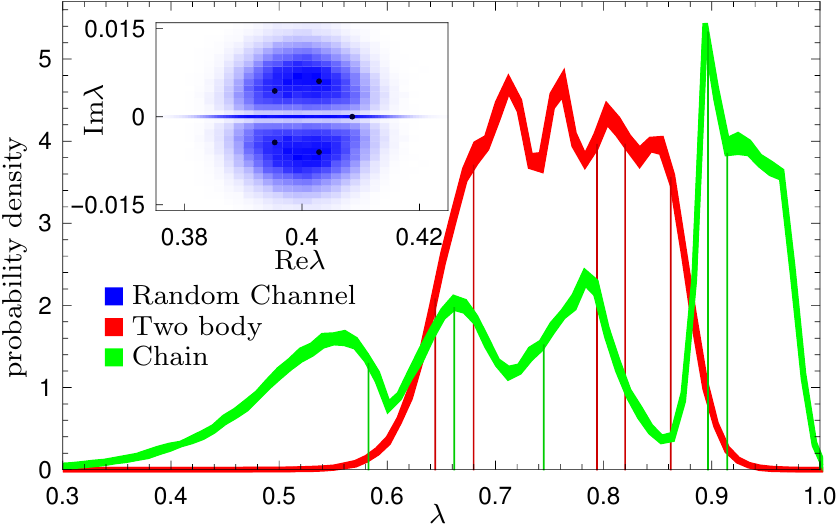}
\caption{
Probability density of the spectrum (aside from the 1s') for random FM,
two-body, 
and one-dimensional chains.
The calculation is done for six qubits, $\gamma=\text{diag}\{5,1\}$, and
$10^5$ realizations.
For 2-body, and 1-d chains we divided the data into 10 groups to obtain a
standard deviation, which is plotted as the thickness of the curves. An example
of a spectrum 
for a single realization is included as darker dots (inset) or lines (figure). 
\label{fig:distribution}}
\end{figure} \subsection{Spectrum} 
We now examine some features of the spectrum of FM. 
In the spectrum we observe exactly $C(d^2+n-1, n)$ ones, one for each block, confirming
what we obtained using the Perron-Frobenius theorem. 
We also note that changing the probability of error, while maintaining the 
relative probabilities of other permutations fixed, just rescales the 
spectrum since
$\text{Spec}(p \openone + (1-p) \mcF )=p + (1-p)\text{Spec}(\mcF )$.
It must also be noted that for general channels, the spectrum is complex. However, for channels involving only swaps
[e.g. \eref{eq:fuzzy_two_body} and \eref{eq:fuzzy_chain}] the spectrum 
is real as the operator is symmetric. 

To conclude, we analyze the spectrum for generic situations (see
also~\cite{BRUZDA2009320} and \cite{prosenrandomlindblad}). To do so, we define
random measuring devices characterized by the channels
\begin{equation}
\mcF_{\text{ran}}[\rho] = 
    p \rho + (1-p)\sum_{P\in \mcS} p_P 
P[\rho],
\label{eq:random}
\end{equation}
where $p$ describes the probability of
doing the correct measurement and the $p_P$ values are chosen uniformly and
normalized such that $\sum_{P\in \mcS} p_P=1$. Random two-body and
one-dimensional FMs are similarly defined, via \eref{eq:fuzzy_two_body} and
\eref{eq:fuzzy_chain} respectively.  For a random measuring device,
$\mcF_{\text{ran}}$, aside from the 1s eigenvalues discussed above, we observe
self-averaging of the eigenvalues around $1-p$, see \fref{fig:distribution}.
This implies a spectral gap of approximately $p$. For random 2-body and
1-dimensional FM there is a richer structure and, if existent,
selfaveraging is much slower, see \fref{fig:distribution}.

 \section{Conclusions}
In this work we have addressed the theoretical characterization of quantum systems
from imperfect measurements and provided a solid framework in the language of
quantum channels to build fuzzy and coarse grained descriptions of them.  The
former not only is a critical step for the latter, but furnishes a quantum
channel scheme to address particle-indexing experimental errors. Fuzzy maps do
not reduce the size of the system, yet they diminish non-symmetric correlations
with respect to the permutation group generated by the elements defining them.
Coarse graining maps, on the other hand, combine contributions from every
particle in the system into a reduced state, effectively lowering the number of
particles and giving rise to coarser particles. In this paper, besides their
general definitions, we explicitly construct maps corresponding to typical
situations encountered in present state-of-the-art experiments in many-body
physics, e.g., two-body errors and low resolution apparatuses. 

In order to illustrate the use of the proposed maps we apply them to describe
entanglement waves in a spin-$1/2$ $XX$-chain. As expected, in the portrayal
obtained, the quantum correlations in the system are partially concealed,
reflecting a decoherence-like effect due to the loss of information intrinsic
to FMs and CG measurements. Such an outcome may be of relevance for some of the
results reported in \cite{Sherson:2010hg}.

Though it is well known that in practice we cannot reconstruct the whole state
of a many-body quantum system, our framework offers tools to compute the
accessible state set in a wide number of scenarios. In the last part of the
paper we have carefully studied the symmetries and spectra of the FM and CG
channels, and identified the set of invariant states as the completely symmetric
states.  Remarkably, considering nonsymmetric states, we have shown that the
volume of the accessible states contracts at a doubly exponential rate when
fuzzy-like noise is present, signaling the fragility of quantum correlation
under imperfect detection.

\par
\section{Acknowledgments}
Conversations with F. de Melo, T. H. Seligman, J. D. Urbina,
and A. Diaz-Ruelas are also acknowledged. 
Support by projects CONACyT 285754 and UNAM-PAPIIT IG100518 and IG101421 
is also acknowledged. 
\appendix
\section{Connection between ket-bras in $\mcB_\gamma$} \label{app:connectiongamma}
In this section we provide the connection between elements of the computational
basis lying in the same invariant sector of $\mcF$. This result is needed to
prove that sectors $\mcB_\gamma$ (set of bounded operators spanned by the
elements of the computational basis with matrix eigenvalue $\gamma$; see text
for further details) are irreducible in the generic case.
\begin{theorem}[Connection of elements in $\mcB_\gamma$]
Let $\proj{\vec x}{\vec y}, \ \ \proj{\vec x'}{\vec y'} \in \mcB_\gamma$ where
$\ket{\vec x}$, $\ket{\vec y}$, $\ket{\vec x'}$ and $\ket{\vec y'}$ are
elements of the computational basis, and $\mcB_\gamma$ is the support of
$\mcF_\gamma$; then  $\exists P\in \mcS$ such that
\begin{equation}
\proj{\vec x'}{\vec y'}=P\proj{\vec x}{\vec y}P^\dagger.
\end{equation}
\end{theorem}
\begin{proof}
Let us show that every $\proj{\vec x}{\vec y}\in \mcB_\gamma$
can be written as a permutation of a reference ket-bra, \ie{}, $\proj{\vec
x}{\vec y}=P[\proj{\vec x_\text{r}^\gamma}{\vec y_\text{r}^\gamma}]$, where the subscript
``r'' stands for
\textit{reference}. 
Let us now rewrite 
\begin{equation}
\proj{\vec x}{\vec y} = 
\proj{x_1}{y_1} \otimes \dots
\otimes 
\proj{x_N}{y_N}.
\end{equation}
By definition, there will be $\gamma_{i,j}$ instances of the single particle
operator $|i\>\<j|$. One can then order the particles such that the $\gamma_{0,0}$
operators $|0\>\<0|$ are firsts, the $\gamma_{0,1}$
operators $|0\>\<1|$ are next, and so on, until the $\gamma_{d-1,d-1}$
operators $|d-1\>\<d-1|$.
This is the reference operator. In other words, 
\begin{multline}
\proj{\vec x_\text{r}^\gamma}{\vec y_\text{r}^\gamma} = 
|0\>\<0|^{\otimes \gamma_{0,0}} \otimes 
|0\>\<1|^{\otimes \gamma_{0,1}}\otimes 
\\
\dots \otimes |d-1\>\<d-1|^{\otimes \gamma_{d-1,d-1}}.
\end{multline}
Such a permutation will be called $P$. Let $\proj{\vec x'}{\vec y'}\in
\mcB_\gamma$ be another element of the computational 
basis characterized by the same $\gamma$, and let $P'$ be the corresponding permutation. 
Then, since 
\begin{align}
\proj{\vec x}{\vec y}   &= 
    P\proj{\vec x_\text{r}^\gamma}{\vec y_\text{r}^\gamma} P^\dagger \\
\proj{\vec x'}{\vec y'} &= 
    P'\proj{\vec x_\text{r}^\gamma}{\vec y_\text{r}^\gamma} P'^\dagger,
\end{align}
we have that 
\begin{align}
\proj{\vec x}{\vec y} &= P (P')^\dagger \proj{\vec x'}{\vec y'} P' P^\dagger \\
&= P''  \proj{\vec x'}{\vec y'} (P'')^\dagger,
\end{align}
with $P''=P (P')^\dagger$. 
This completes the proof.
\end{proof}
\section{Invariant states} \label{app:invariant}
In the text, using the symmetries of the generic fuzzy map and the
Perron-Frobenius theorem, we have proven that a matrix $\Delta$ is invariant
under generic FM if and only if it is invariant under any permutation. Here we
give a theorem that holds only for Hermitian positive-definite operators, but
that holds for nongeneric maps.

\begin{theorem}[Invariant Hermitian matrices]
Let $\Delta$ be a bounded positive-definite operator, and $\mcF$ an FM 
defined by a probability vector $\vec p$  whose non-zero entries multiply terms
corresponding to a set of permutations $\mcA$. Then,
\begin{equation}
P[\Delta]=\Delta \ \ \forall P\in \text{G}_\mcA,
\end{equation}
where $\text{G}_\mcA$ is the permutation group generated by $\mcA$.
\end{theorem}
\begin{proof}
Note that if $\mcF[\Delta]=\Delta$, then
$\tr[\Delta^2]=\tr\left[\mcF[\Delta]^2\right]$ holds. For density matrices this
equality simply means that the purity is preserved. Developing the expression
we get
\begin{align}
\tr[\Delta^2]&=\tr\left[\mcF[\Delta]^2\right]\nonumber\\
&=\tr\left[\left(\sum_{ P \in \mcA}p_{ P} P \Delta P^{\dagger}\right)^2\right]\nonumber\\
&=\sum_{ P \in \mcA'} q_{ P} \tr \left[  P \Delta  P^\dagger{} \Delta \right],
\label{eq:convex_purity_sum}
\end{align}
where the $q_P$s are quadratic functions of the $p_P$s and $\mcA'$ is a set of
permutations generated by pairwise concatenations of elements in $\mcA$. Note
that the $q_P$s form a probability distribution. This follows from the fact
that each $q_P$ is a sum of elements of a product distribution (see the second
inequality). Thus, the sum in~\eref{eq:convex_purity_sum} is a convex
combination of traces.
Observe that $\mcA \subseteq \mcA'$, where the equality holds only if $\mcA$ is a permutation group.
Using the Bellman inequalities~\cite{Bellman1980,Zhou2014} we have $\tr[
P\Delta P^{\dagger}\Delta]\leq \tr[\Delta^2]$, hence the convex combination
in~\eref{eq:convex_purity_sum} equals $\tr \Delta^2$ only if $\tr
[\Delta^2]=\tr[ P \Delta P^\dagger{} \Delta]$. 

Consider now the trace norm,
\begin{align*}
||\Delta-P {\Delta}P^{\dagger}||_2^2
&=\tr\left[\left|\Delta-P\Delta P^\dagger{}\right|^2\right]\\
&=\tr\left[\left(\Delta-P\Delta P^\dagger{}\right)^2\right]\\
&=\tr\left[\Delta^2-\Delta P \Delta P^\dagger{}-P \Delta P^\dagger{} \Delta\right.\\
& \qquad \qquad   \left. + P\Delta^2 P^\dagger{}\right]\\
&=2\left(\tr \left[\Delta^2\right]-\tr \left[P \Delta P^\dagger{} \Delta \right]\right)\\
&=0.
\end{align*}
The second equality follows from the Hermiticity of $\Delta$. Therefore, by the
properties of norms, $\Delta=P\Delta P^\dagger$ for all $P\in \mathcal{A}'$.

Note that if $\Delta$ is invariant under permutations in $\mcA'$, then it is
invariant under the permutation group generated by it, which in turn is the
same group generated by $\mcA$, \ie{}, $\text{G}_\mcA$.
\end{proof}

 \bibliographystyle{unsrt}

\begin{thebibliography}{10}

\bibitem{Zurek:99827}
John~Archibald Wheeler and Wojciech~H Zurek, editors.
\newblock {\em {Quantum theory and measurement}}.
\newblock Princeton University Press, Princeton, NJ, 1983.

\bibitem{Ott:2016dc}
Herwig Ott.
\newblock Single atom detection in ultracold quantum gases: a review of current
  progress.
\newblock {\em Rep. Prog. Phys.}, 79(5):054401, 2016.

\bibitem{RevModPhys.91.021001}
Dmitry~A. Abanin, Ehud Altman, Immanuel Bloch, and Maksym Serbyn.
\newblock Colloquium: Many-body localization, thermalization, and entanglement.
\newblock {\em Rev. Mod. Phys.}, 91:021001, 2019.

\bibitem{Bloch:2018kb}
Immanuel Bloch.
\newblock {Quantum simulations come of age}.
\newblock {\em Nat. Phys.}, 14:1159--1162, 2018.

\bibitem{Wise09}
Howard~M Wiseman and Gerard~J Milburn.
\newblock {\em {Quantum Measurement and Control}}.
\newblock Cambridge University Press, New York, 2010.

\bibitem{Jacobs:2014ui}
Kurt Jacobs.
\newblock {\em Quantum Measurement Theory and its Applications}.
\newblock Cambridge University Press, 2014.

\bibitem{Haffner:2005hs}
H~H{\"a}ffner, W~H{\"a}nsel, C~F Roos, J~Benhelm, D~Chek-al kar, M~Chwalla,
  T~K{\"o}rber, U~D Rapol, M~Riebe, P~O Schmidt, C~Becher, O~G{\"u}hne,
  W~D{\"u}r, and R~Blatt.
\newblock {Scalable multiparticle entanglement of trapped ions}.
\newblock {\em Nature}, 438(7068):643--646, 2005.

\bibitem{Preskill:2018gt}
John Preskill.
\newblock Quantum {C}omputing in the {NISQ} era and beyond.
\newblock {\em {Quantum}}, 2:79, 2018.

\bibitem{peres1995quantum}
A.~Peres.
\newblock {\em Quantum Theory: Concepts and Methods}.
\newblock Fundamental Theories of Physics. Springer Netherlands, 1995.

\bibitem{PhysRevD.20.384}
M.~B. Mensky.
\newblock Quantum restrictions for continuous observation of an oscillator.
\newblock {\em Phys. Rev. D}, 20:384--387, 1979.

\bibitem{GUDDER200518}
Stan Gudder.
\newblock Non-disturbance for fuzzy quantum measurements.
\newblock {\em Fuzzy Sets and Systems}, 155(1):18 -- 25, 2005.

\bibitem{carmeli2008}
Claudio Carmeli, Teiko Heinonen, and alessandro toigo.
\newblock Why unsharp observables?
\newblock {\em Int. J. Theo. Phys.}, 47(1):81--89, 2008.

\bibitem{Busch2010}
Paul Busch and Gregg Jaeger.
\newblock Unsharp quantum reality.
\newblock {\em Found. Phys.}, 40(9):1341--1367, 2010.

\bibitem{Brukner2007}
Johannes Kofler and {\v{C}}aslav Brukner.
\newblock Classical world arising out of quantum physics under the restriction
  of coarse-grained measurements.
\newblock {\em Phys. Rev. Lett.}, 99:180403, 2007.

\bibitem{Brukner2008}
Johannes Kofler and {\v{C}}aslav Brukner.
\newblock Conditions for quantum violation of macroscopic realism.
\newblock {\em Phys. Rev. Lett.}, 101:090403, 2008.

\bibitem{Duarte2017}
Cristhiano Duarte, Gabriel~Dias Carvalho, Nadja~K. Bernardes, and Fernando
  de~Melo.
\newblock Emerging dynamics arising from coarse-grained quantum systems.
\newblock {\em Phys. Rev. A}, 96:032113, 2017.

\bibitem{PhysRevA.100.022334}
Pedro Silva~Correia and Fernando de~Melo.
\newblock Spin-entanglement wave in a coarse-grained optical lattice.
\newblock {\em Phys. Rev. A}, 100:022334, 2019.

\bibitem{Carvalho:2020fo}
Gabriel~Dias Carvalho and Pedro~Silva Correia.
\newblock Decay of quantumness in a measurement process: Action of a
  coarse-graining channel.
\newblock {\em Phys. Rev. A}, 102:032217, 2020.

\bibitem{Note1}
Consider $\left \{A_i \right \}_i$ to be a set of observables tomographically
  complete, suffering from the fuzzy noise. Then the measurable averages are
  {$\mathop {\protect \mathrm {Tr}}\nolimits ((p A_{i} + (1-p) S A_{i} S) \rho
  ) = \langle A_i \rangle _{\protect \mathcal {F}_{2\protect \text {p}}[\rho
  ]}$} for all $A_i$. Since {$\langle A_i \rangle _{\protect \mathcal
  {F}_{2\protect \text {p}}[\rho ]}$} are the components of $\protect \mathcal
  {F}_{2\protect \text {p}}[\rho ]$ in the complete operator basis $\left \{A_i
  \right \}$, then $\protect \mathcal {F}_{2\protect \text {p}}[\rho ]$ is the
  state accessible tomographically.

\bibitem{Fukuhara:2013hq}
Takeshi Fukuhara, Adrian Kantian, Manuel Endres, Marc Cheneau, Peter
  Schau{\ss}, Sebastian Hild, David Bellem, Ulrich Schollw{\"o}ck, Thierry
  Giamarchi, Christian Gross, Immanuel Bloch, and Stefan Kuhr.
\newblock {Quantum dynamics of a mobile spin impurity}.
\newblock {\em Nat. Phys.}, 9:235--241, 2013.

\bibitem{Fukuhara:2015ec}
Takeshi Fukuhara, Sebastian Hild, Johannes Zeiher, Peter Schau\ss{}, Immanuel
  Bloch, Manuel Endres, and Christian Gross.
\newblock Spatially resolved detection of a spin-entanglement wave in a
  bose-hubbard chain.
\newblock {\em Phys. Rev. Lett.}, 115:035302, 2015.

\bibitem{Sherson:2010hg}
Jacob~F Sherson, Christof Weitenberg, Manuel Endres, Marc Cheneau, Immanuel
  Bloch, and Stefan Kuhr.
\newblock {Single-atom-resolved fluorescence imaging of an atomic Mott
  insulator}.
\newblock {\em Nature}, 467(7311):68--72, 2010.

\bibitem{Bakr:2010gd}
W~S Bakr, A~Peng, M~E Tai, R~Ma, J~Simon, and {2010}.
\newblock {Probing the superfluid{\textendash}to{\textendash}Mott insulator
  transition at the single-atom level}.
\newblock {\em Science}, 239:547--550, 2010.

\bibitem{Subrahmanyam:2004eo}
V~Subrahmanyam.
\newblock {Entanglement dynamics and quantum-state transport in spin chains}.
\newblock {\em Phys. Rev. A}, 69(3):034304, 2004.

\bibitem{Amico:2004ck}
Luigi Amico, Andreas Osterloh, Francesco Plastina, Rosario Fazio, and
  G.~Massimo~Palma.
\newblock Dynamics of entanglement in one-dimensional spin systems.
\newblock {\em Phys. Rev. A}, 69:022304, 2004.

\bibitem{Konno:2005da}
Norio Konno.
\newblock Limit theorem for continuous-time quantum walk on the line.
\newblock {\em Phys. Rev. E}, 72:026113, 2005.

\bibitem{Mazza:2015km}
Leonardo Mazza, Davide Rossini, Rosario Fazio, and Manuel Endres.
\newblock Detecting two-site spin-entanglement in many-body systems with local
  particle-number fluctuations.
\newblock {\em New J. Phys}, 17(1):013015, 2015.

\bibitem{starsbars}
W~Feller.
\newblock {\em {An Introduction to Probability Theory and its Applications}},
  volume~1.
\newblock John Wiley and Sons, Inc., 1968.

\bibitem{Meyer2001}
Carl~D. Meyer.
\newblock {\em {Matrix Analysis and Applied Linear Algebra}}.
\newblock Society for Industrial and Applied Mathematics, 2001.

\bibitem{perron}
Barry Simon.
\newblock {\em {Real Analysis}}.
\newblock American Mathematical Society, Providence, Rhode Island, 2015.

\bibitem{ergodic}
D~Burgarth, G~Chiribella, V~Giovannetti, P~Perinotti, and K~Yuasa.
\newblock {Ergodic and mixing quantum channels in finite dimensions}.
\newblock {\em New J. Phys.}, 15(7):073045, 2013.

\bibitem{Banaszek_2013}
K.~Banaszek, M.~Cramer, and D.~Gross.
\newblock Focus on quantum tomography.
\newblock {\em New J. Phys}, 15(12):125020, 2013.

\bibitem{Flammia_2012}
Steven~T Flammia, David Gross, Yi-Kai Liu, and Jens Eisert.
\newblock Quantum tomography via compressed sensing: error bounds, sample
  complexity and efficient estimators.
\newblock {\em New J. Phys.}, 14(9):095022, 2012.

\bibitem{Orus2019}
Rom{\'a}n Or{\'u}s.
\newblock Tensor networks for complex quantum systems.
\newblock {\em Nature Reviews Physics}, 1(9):538--550, 2019.

\bibitem{2009arXiv0905.0669P}
Carlos Pineda, Thomas Barthel, and Jens Eisert.
\newblock Unitary circuits for strongly correlated fermions.
\newblock {\em Phys. Rev. A}, 81(5):050303(R), 2010.

\bibitem{UHLMANN1976273}
A.~Uhlmann.
\newblock The ``transition probability'' in the state space of a $*$-algebra.
\newblock {\em Rep. Math. Phys.}, 9(2):273 -- 279, 1976.

\bibitem{geometry2017}
Ingemar Bengtsson and Karol Zyczkowski.
\newblock {\em Geometry of Quantum States: An Introduction to Quantum
  Entanglement}.
\newblock Cambridge University Press, 2006.

\bibitem{Horodecki}
R.~Horodecki, P.~Horodecki, M.~Horodecki, and K.~Horodecki.
\newblock {Quantum entanglement}.
\newblock {\em Rev. Mod. Phys.}, 81(2):865--942, 2009.

\bibitem{KUMAR20163044}
Asutosh Kumar.
\newblock Conditions for monogamy of quantum correlations in multipartite
  systems.
\newblock {\em Phys. Lett. A}, 380(38):3044 -- 3050, 2016.

\bibitem{BRUZDA2009320}
Wojciech Bruzda, Valerio Cappellini, Hans-Jürgen Sommers, and Karol
  Życzkowski.
\newblock Random quantum operations.
\newblock {\em Phys. Lett. A}, 373(3):320 -- 324, 2009.

\bibitem{prosenrandomlindblad}
Lucas S{\'{a}}, Pedro Ribeiro, and Toma{\v{z}} Prosen.
\newblock Spectral and steady-state properties of random liouvillians.
\newblock {\em J. Phys. A}, 53(30):305303, 2020.

\bibitem{Bellman1980}
Richard Bellman.
\newblock {\em {Some Inequalities for Positive Definite Matrices}}, pages
  89--90.
\newblock Birkh{\"{a}}user Basel, Basel, 1980.

\bibitem{Zhou2014}
Houqing Zhou.
\newblock {On some trace inequalities for positive definite Hermitian
  matrices}.
\newblock {\em J Inequal Appl}, 2014(1):64, 2014.

\end{thebibliography}

\end{document}